\documentclass[conference]{IEEEtran}

\usepackage[T1]{fontenc}  
\usepackage{amsmath,epsfig,amssymb,amsfonts,amstext,verbatim,amsopn,cite,subfigure}
\usepackage{multirow,multicol,lipsum,xfrac}
\usepackage{amsthm}
\usepackage{mathtools,amsthm}
\usepackage{url}
\usepackage{amsfonts}
\usepackage{epsfig}
\usepackage[font=small]{caption}
\usepackage{psfrag}	
\usepackage[Algorithm,ruled]{algorithm}
\usepackage{algorithmicx}
\usepackage{algpseudocode}
\usepackage{pifont}
\usepackage[utf8]{inputenc}
\usepackage[nolist]{acronym}
\usepackage{paralist}
\usepackage{enumitem}
\usepackage{dsfont}
\usepackage[process=auto]{pstool}
\usepackage{float}

\usepackage{tikz,pgfplots}
\usetikzlibrary{shapes.misc,arrows,positioning}

\newcommand{\setR}{\mathds{R}}

\newcommand{\setZ}{\mathds{Z}}

\newcommand{\setI}{\mathds{I}}

\newcommand{\Ex}[2]{ \mathbb{E}_{#2} \left\lbrace #1 \right\rbrace }

\newcommand{\E}{\mathrm{E}}
\newcommand{\B}{\mathrm{B}}

\newcommand{\tK}{{\tilde{K}}}
\newcommand{\bk}{\mathbf{k}}
\newcommand{\bts}{\tilde{\mathbf{s}}}

\newcommand{\rmS}{\mathrm{S}}

\newcommand{\maz}{\mathcal{Z}}

\newcommand{\mav}{\mathcal{V}}

\newcommand{\mai}{\mathcal{I}}

\newcommand{\mac}{\mathcal{C}}

\newcommand{\mal}{\mathcal{L}}
\newcommand{\mPi}{\boldsymbol{\Pi}}

\newcommand{\bxx}{\mathbf{x}}

\newcommand{\bww}{\mathbf{w}}
\newcommand{\buu}{\mathbf{u}}

\newcommand{\byy}{\mathbf{y}}

\newcommand{\rr}{\mathrm{r}}

\newcommand{\rms}{\mathrm{s}}

\newcommand{\xx}{\mathrm{x}}
\newcommand{\sgn}[1]{\mathrm{sgn} \left( #1 \right)}
\newcommand{\kk}{\mathrm{k}}

\newcommand{\set}[1]{\left\lbrace#1\right\rbrace}

\newcommand{\brc}[1]{\left( #1 \right) }
\newcommand{\inner}[1]{\left\langle #1 \right\rangle }

\newcommand{\dbc}[1]{\left[ #1 \right] }

\newcommand{\bs}{{\mathbf{s}}}
\newcommand{\br}{{\mathbf{r}}}

\newcommand{\bv}{{\mathbf{v}}}

\newcommand{\trp}{\mathsf{T}}

\newcommand{\mA}{\mathbf{A}}

\newcommand{\md}{\mathrm{D}}

\newcommand{\norm}[1]{\lVert #1 \rVert}

\newcommand{\abs}[1]{\lvert #1 \rvert}

\newtheoremstyle{mystyle}
{}
{}
{\it}
{}
{\bfseries}
{:}
{ }
{}

\theoremstyle{mystyle}

\newtheorem{definition}{Definition}
\newtheorem{proposition}{Proposition}

\newtheorem{result}{Result}

\newcounter{bar}

\begin{document}

\begin{acronym}
	\acro{mimo}[MIMO]{multiple-input multiple-output}
	\acro{csi}[CSI]{channel state information}
	\acro{awgn}[AWGN]{additive white Gaussian noise}
	\acro{iid}[i.i.d.]{independent and identically distributed}
	\acro{uts}[UTs]{user terminals}
	\acro{bs}[BS]{base station}
	\acro{tas}[TAS]{transmit antenna selection}
	\acro{glse}[GLSE]{generalized least square error}
	\acro{rhs}[r.h.s.]{right hand side}
	\acro{lhs}[l.h.s.]{left hand side}
	\acro{wrt}[w.r.t.]{with respect to}
	\acro{rs}[RS]{replica symmetry}
	\acro{rsb}[RSB]{replica symmetry breaking}
	\acro{np}[NP]{non-deterministic polynomial-time}
	\acro{papr}[PAPR]{peak-to-average power ratio}
	\acro{rzf}[RZF]{regularized zero forcing}
	\acro{snr}[SNR]{signal-to-noise ratio}
	\acro{rem}[REM]{random energy model}
	\acro{mf}[MF]{matched filtering}
	\acro{gamp}[GAMP]{generalized AMP}
	\acro{amp}[AMP]{approximate message passing}
	\acro{vamp}[VAMP]{vector AMP}
	\acro{map}[MAP]{maximum-a-posterior}
	\acro{ml}[ML]{maximum likelihood}
	\acro{mmse}[MMSE]{minimum mean squared error}
	\acro{ap}[AP]{average power}
	\acro{ldgm}[LDGM]{low density generator matrix}
	\acro{tdd}[TDD]{time division duplexing}
	\acro{rss}[RSS]{residual sum of squares}
	\acro{rls}[RLS]{regularized least-squares}
	\acro{ls}[LS]{least-squares}
	\acro{erp}[ERP]{encryption redundancy parameter}
	\acro{ra}[RA]{reflect-array}
	\acro{ta}[TA]{transmit-array}
	\acro{bsc}[BSC]{binary symmetric channel}
\end{acronym}

\title{Secure Coding via Gaussian Random Fields}

\author{
	\IEEEauthorblockN{
		Ali Bereyhi\IEEEauthorrefmark{1},
		Bruno Loureiro\IEEEauthorrefmark{2},
		Florent Krzakala\IEEEauthorrefmark{2},
		Ralf R. M\"uller\IEEEauthorrefmark{1}, and 
		Hermann Schulz-Baldes\IEEEauthorrefmark{3}
	}
	\IEEEauthorblockA{
		\IEEEauthorrefmark{1}Institute for Digital Communications, Friedrich-Alexander Universit\"at Erlangen-N\"urnberg (FAU)\\
		\IEEEauthorrefmark{2}Information, Learning and Physics Lab, École Polytechnique Fédérale de Lausanne (EPFL)\\
		\IEEEauthorrefmark{3}Mathematical Physics and Operator Algebras Group, Department of Mathematics, FAU\\
		ali.bereyhi@fau.de, bruno.loureiro@epfl.ch, florent.krzakala@epfl.ch, ralf.r.mueller@fau.de, schuba@mi.uni-erlangen.de
		\thanks{This work has been accepted for presentation in 2022 IEEE International Symposium on Information Theory (ISIT) in Espoo, Finland. The link to the final version in the Proceedings of ISIT will be available later.}
		\thanks{This work was supported by the Emerging Talents Initiative (ETI) at  FAU.}
	}
}

\IEEEoverridecommandlockouts
\maketitle

\begin{abstract}
Inverse probability problems whose generative models are given by \textit{strictly nonlinear} Gaussian random fields show the \textit{all-or-nothing} behavior: There exists a critical rate at which Bayesian inference exhibits a \textit{phase transition}. Below this rate, the optimal Bayesian estimator recovers the data perfectly, and above it the recovered data becomes uncorrelated.~This~study uses the replica method from the theory of spin glasses to show~that this critical rate is the channel capacity. This interesting finding has a particular application to the problem of secure transmission: A {strictly} nonlinear Gaussian random field along with random binning can be used to securely encode a confidential message in a wiretap channel. Our large-system characterization demonstrates that this secure coding scheme asymptotically achieves the secrecy capacity of the Gaussian wiretap channel.
\end{abstract}


\IEEEpeerreviewmaketitle

\section{Preliminaries}
Unlike linear models \cite{barbier2016mutual,barbier2018mutual,bereyhi4,bereyhi1Extension,barbier2020mutual,bereyhi2020thesis,barbier2021performance}, nonlinear models are sparsely studied in information theory. This follows from the fact that these models are more complex and seemly have a narrower scope of applications. They are however rather popular in physics; for instance, in the theory of spin glasses \cite{edwards1975theory,mezard1987spin,panchenko2013sherrington,talagrand2006free}. Invoking connections between statistical mechanics~and~information theory, a few studies establish known nonlinear models in the theory of spin glasses to address information-theoretic problems, such as channel coding and data encryption \cite{sourlas1989spin,kabashima1999statistical,fyodorov2019spin}. This study brings light to the secrecy potential of nonlinear models.

\subsection{Contributions and Related Work}
The main contribution of this study is as follows: 
We show that the secrecy capacity of the Gaussian wiretap channel given in \cite{leung1978gaussian} is achieved by employing a \textit{strictly nonlinear} Gaussian random field as the encoder. 
The motivation comes from a recent work by Fyodorov \cite{fyodorov2019spin}, where a Gaussian random field is used for signal encryption. In particular, Fyodorov shows that encryption via a \textit{purely} quadratic Gaussian random field shows an asymptotic \textit{phase transition} at a threshold \ac{snr}, below which recovery via the method of least-squares becomes uncorrelated. Our initial investigations in \cite[Section~V]{bereyhi2022SecureLearning} shows that combining Fyodorov's encryption with a simple sphere coding technique can achieve a perfect secrecy rate close to the secret capacity of the wiretap channel. This interesting finding motivates a new coding scheme to achieve \textit{reliable} and \textit{perfectly secure} transmission over the Gaussian wiretap channel. 

Intuitively, the proposed encoder can be seen as a combination of Sourlas' coding \cite{sourlas1989spin} with the random linear binning technique \cite{wyner1974recent}. Decoding is performed via a standard Bayesian decoder. Using the replica method, we show that the Bayesian decoder asymptotically exhibits the \textit{all-or-noting} behavior \cite{barbier2020all,niles2021all,maillard2020phase} when a \textit{strictly nonlinear} Gaussian random field is used as the encoder. This finding is then used to show that the proposed coding scheme asymptotically achieves the secrecy capacity of the Gaussian wiretap channel.
\subsection{Notation}
Scalars, vectors and matrices are represented with non-bold, bold lower-case and bold upper-case letters, respectively. The transposed of $\mA$ is indicated by $\mA^{\trp}$. The Euclidean norm of $\bxx$ is denoted by $\norm{\bxx}$. For $K$-dimensional vectors $\bxx$ and $\byy$, we define the \textit{normalized inner-product} as
\begin{align}
	\inner{\bxx;\byy} = \frac{\bxx^\trp \byy}{K}. \label{eq:inner}
\end{align}
The notation $\log\brc{\cdot}$ indicates the \textit{natural} logarithm, and $\Ex{\cdot}{}$ denotes mathematical expectation. The real axis is shown by $\setR$. For sake of brevity, $\set{1,\ldots,N}$  and $\set{N_0,N_0+1,\ldots,N_1}$ are abbreviated as $\dbc{N}$ and $\dbc{N_0:N_1}$, respectively. The capacity of a real Gaussian channel with \ac{snr} $x$ is denoted~by
\begin{align}
	\mac\brc{x} = \frac{1}{2} \log\brc{1+x}.
\end{align}
%

\subsection{Gaussian Random Fields}
Gaussian random fields are key components in this paper. We hence define them at this point: In a nutshell, a Gaussian random field is a randomized mapping whose output entries are Gaussian. The mapping is not necessarily linear, i.e. a random matrix, and can be of higher orders. It can hence be observed as an extension of Gaussian random matrices. The exact definition is given below.
\begin{definition}
	The mapping $\mav\brc{\cdot}: \setR^K \mapsto \setR^N$ is a Gaussian random field with covariance function $\Phi\brc\cdot : \setR\mapsto \setR$, if for any pair of vectors $\bs_1,\bs_2 \in \setR^K$, the entries of 
	\begin{align}
		\mav\brc{\bs_i} = \dbc{ \mav_1\brc{\bs_i} , \ldots , \mav_N\brc{\bs_i} }^\trp,
	\end{align}
for $i\in\set{1,2}$, are distributed Gaussian and satisfy
	\begin{align}
		\Ex{\mav_m\brc{\bs_1} \mav_n\brc{\bs_2} }{ } = \setI\set{ m = n} \Phi\brc{ \inner{\bs_1;\bs_2} } ,
	\end{align}
	for $n,m\in\dbc{N}$, with $\setI\set{ m \hspace*{-.5mm}= \hspace*{-.5mm} n}$ being the indicator function.
\end{definition}
A basic example of a Gaussian random field is the \textit{linear} field with covariance function $\Phi\brc{u} = u $. The Gaussian field in this example can be represented as $\mav\brc{\bs} = \mA \bs$, where $\mA\in\setR^{N\times K}$ is an \ac{iid} Gaussian matrix whose entries are zero-mean with variance $1/K$. More examples can be found in \cite{bereyhi2022SecureLearning}.


\section{Statement of Main Result}
Consider a Gaussian wiretap channel in which a transmitter intends to securely send a message of $K=NR$ bits to a legitimate receiver over an \ac{awgn} channel by $N$ uses of the channel. An eavesdropper overhears the transmitted signal through an independent \ac{awgn} channel.

Let $M\in \dbc{2^{NR}}$ denote the secret message. The transmitter encodes $M$ via a \textit{secure} encoder %
	$f_N \brc\cdot : \dbc{2^{NR}}\mapsto \setR^N$ %
into the codeword $\bxx = \dbc{\xx_1, \ldots , \xx_N}^\trp$, such that
\begin{align}
	\frac{1}{N} \sum_{n=1}^N \Ex{\abs{\xx_n}^2}{} \leq P,
\end{align}
for some \textit{average transmit power} $P$. It then transmits the codeword over the \ac{awgn} channel using $N$~subsequent~transmission time intervals. The legitimate receiver therefore receives 
\begin{align}
	\byy = \bxx + \bww_\B,
\end{align}
for an \ac{iid} Gaussian noise vector $\bww_\B$ whose entries are zero-mean with variance $\sigma_\B^2$. The legitimate receiver employs the decoder  %
	$g_N \brc\cdot : \setR^N \mapsto \dbc{2^{NR}}$  %
to estimate the secret message, i.e., $\hat{M} = g_N\brc{\byy}$.

The eavesdropper overhears the transmitted sequence $\bxx$ over an independent \ac{awgn} channel and receives 
\begin{align}
	\byy_\E = \bxx + \bww_\E,
\end{align}
for \ac{iid} Gaussian noise $\bww_\E$ whose entries are zero-mean with variance $\sigma_\E^2$. It is further assumed that the encoder and decoder are publicly known to all parties.

\subsection{Reliable and Secure Transmission}
The secret message $M$ contains $K = NR$ information bits and is transmitted within $N$ channel uses. It is hence concluded that the \textit{transmission rate} in this setting is $R$ bits per channel use. We now focus on a sequence of encoder-decoder pairs indexed by $N$ which transmit secret messages with rate $R$. The transmission is said to be \textit{reliable} and \textit{secure} with this encoder-decoder sequence if the following conditions are satisfied \cite{liang2009information}:
\begin{itemize}
	\item The recovery error $\Pr\{ M \neq \hat{M} \}$ tends to zero as~the~dimension $N$ grows 
	large.
	\item The information leakage to eavesdropper, i.e., 
	\begin{align}
		L_N = \frac{1}{N} I \brc{ M ; \byy_\E } ,
	\end{align}
 	tends to zero as the dimension $N$ grows large. 
\end{itemize}


The \textit{secrecy capacity} of the wiretap channel is defined as the maximum reliable and secure transmission rate. Using the random binning approach of Wyner \cite{wyner1975wire}, the secrecy capacity of this Gaussian wiretap channel is shown to be \cite{leung1978gaussian}
\begin{align}
	C_{\rm S} =\dbc{ \mac \brc{\frac{P}{\sigma_\B^2}} - \mac \brc{\frac{P}{\sigma_\E^2}}}^{+} ,\label{eq:sec_Cap}
\end{align}
where $\dbc{x}^+ \coloneqq \max\set{0,x}$. 

\subsection{Secure Coding via Gaussian Random Fields}
\label{sec:Coding}
Invoking Gaussian random fields, we now propose a secure coding scheme. Our main result shows that under some heuristic assumptions, this scheme achieves the secrecy capacity. The proposed scheme uses a \textit{strictly nonlinear} Gaussian random field to securely encode the secret message. The decoder then recovers the message via the optimal Bayesian estimator.

To state the proposed coding scheme, let $\bs\in\set{\pm 1}^K$ be the bipolar representation of the message $M$, i.e., each information bit of $M$ is shown by $\pm 1$. For secure encoding, the transmitter follows the following steps:
\begin{enumerate}
	\item It generates at random
	\begin{align}
		\tK = \left\lceil \frac{N}{\log 2} \mac\brc{\frac{P}{\sigma_\E^2}}\right\rceil \label{eq:tK}
	\end{align}
	\ac{iid} \textit{uniform} bipolar symbols, i.e., 
		$\bk = \dbc{\kk_1, \ldots, \kk_{\tK}}^\trp$.
	\item It adds $\bk$ as a prefix to $\bs$ and finds $\bts$ as 
	$\bts = \mPi {\dbc{\bk^\trp, \bs^\trp}^\trp}$, %
where $\mPi$ %
is a $(K+\tK)\times (K+\tK)$ random permutation\footnote{One can see the connection between $\mPi$ and the random binning technique \cite{wyner1974recent,cover1975proof}; see also \cite{yassaee2014achievability,muramatsu2012construction} and references therein for more details on the random binning technique and its applications.} whose permuted vector $\bts$ has the following property: Let the $\ell$-th \textit{bin} of $\bts$ with size $B = \lceil1+K/\tK \rceil$ be 
\begin{align}
	\dbc{\bts}_\ell =  \set{ \tilde{\rms}_k: k\in \dbc{\brc{\ell-1} B +1 : \ell B } } .
\end{align}
Then, for every choice of $\ell \in [\tK]$, there exists only one symbol of $\bk$ in $\dbc{\bts}_\ell$. 
\item It determines the codeword as %
		$\bxx =  \mav\brc{\bts}$,
	for a Gaussian field $\mav\brc{\cdot}: \setR^{K+\tK} \mapsto \setR^N$ whose covariance function is $\Phi\brc{u} = P u^\lambda$ for some $\lambda \geq 3$.
\end{enumerate}

For recovery, the legitimate receiver applies \ac{mmse} estimation to recover the message. To this end, the legitimate receiver follows the following steps:
\begin{enumerate}
	\item It calculates the sufficient statistic $\tilde{\br}\in \setR^{K+\tK} $ from $\byy_\B$ as %
		$\tilde{\br} = \Ex{\tilde{\bs} \vert \byy_\B , \mav}{}$. %
With uniform prior distribution, the \ac{mmse} estimator is given by
\begin{align}
	\tilde{\br} = \frac{1}{\maz} \sum_{\buu \in \set{\pm1}^{K+\tK} } \buu \exp\set{-\frac{\norm{\byy_\B - \mav\brc{\buu}}^2}{2 \sigma_\B^2}},
\end{align}
for normalization $\maz$ defined as
\begin{align}
{\maz} =  \sum_{\buu \in \set{\pm1}^{K+\tK} } \exp\set{-\frac{\norm{\byy_\B - \mav\brc{\buu}}^2}{2 \sigma_\B^2}}.
\end{align}
\item It sets $\br = \mPi^\trp \tilde{\br}$, finds 
	$\hat{\rms}_k = \sgn{\rr_{\tK+k}}$ %
for $k\in \dbc{K}$ with $\sgn{\cdot}$ being the sign operator,
and finds $\hat{M}$ from $\hat{\bs}$.
\end{enumerate}

\subsection{The Main Result}
The main result of this study states that the proposed coding scheme achieves the secrecy capacity of the wiretap channel. We present this result explicitly below:
\begin{result}
	\label{res:1}
	Consider the sequence $\set{ \brc{f_N ,g_N }: N\in \setZ }$, where $f_N$ and $g_N$ denote the encoder and decoder in Section~\ref{sec:Coding}, respectively. Set $K= N C_{\rm S}/\log 2$, where $C_{\rm S}$ is the secrecy~capacity of the wiretap channel. Then, the transmission via this sequence is asymptotically reliable and secure.
\end{result}

\begin{proof}
	The proof is given in Section~\ref{sec:Proof}.
\end{proof}

Considering Result~\ref{res:1}, it is worth mentioning few remarks:
\begin{itemize}
	\item We first note the constraint on the covariance of $\mav\brc{\cdot}$, i.e.,~$\lambda \geq 3$.~It indicates that the codewords are generated via a mapping whose total number of coefficients grows with the message length at least cubically. In general, using a Gaussian field with covariance function $\Phi\brc{u} = u^\lambda$, the codewords are generated via an order $\lambda$ polynomial from a set of $NK^\lambda$ coefficients.
	\item The proposed scheme is randomized, as the mapping $\mav\brc{\cdot}$ is generated at random. Nevertheless, unlike the classical random coding which generates all $N2^K$ components of the codebook at random, the proposed~scheme~uses~$NK^\lambda$ random coefficients to specify $\mav\brc{\cdot}$. It then~constructs~the codewords from the messages using $\mav\brc{\cdot}$.
	\item The key property of the proposed scheme which leads to achieving the secrecy capacity of the wiretap channel is the \textit{strict non-linearity} of the Gaussian field\footnote{This becomes clear later on.}. In fact, by adding a linear term to the encoding field, one can observe that the information leakage to the eavesdropper does not vanish asymptotically. This is justified by a rather known property of the linear model: \textit{Linear} models \textit{always} carry information about the model parameters\footnote{See \cite[Section~IV]{bereyhi2022SecureLearning} for some discussions in this respect.}.
\end{itemize}


\section{Asymptotics of the Coding Scheme}
The derivation of Result~\ref{res:1} relies on the asymptotic characterization of a class Bayesian algorithms used for unsupervised learning in nonlinear generative models. The detailed derivations are given in the extended manuscript \cite{bereyhi2022SecureLearning}. In the sequel, we state a particular form of the generic result in \cite{bereyhi2022SecureLearning} which describes the asymptotic properties of the decoder when it is employed to decode a message encoded via a Gaussian random field. This result is then utilized to sketch a proof for  Result~\ref{res:1}. 

\subsection{Asymptotic Characterization of the Bayesian Decoder}
\label{sec:vec}
The statistics of the sufficient statistic $\tilde{\br}$ are asymptotically described via the results\footnote{The characterization in \cite{bereyhi2022SecureLearning} invokes the replica method.} of \cite{bereyhi2022SecureLearning}. We illustrate the asymptotic characterization through the following setting:~Consider~a~vector of uniform bipolar symbols $\bs_0 \in \set{\pm 1}^K$ which is mapped via a Gaussian random field $\mav\brc\cdot$ into $\bxx_0\in\setR^N$. The mapped vector $\bxx_0$ is observed as $\byy_0 = \bxx_0 + \bww_0$ for Gaussian noise $\bww_0$ whose entries are zero-mean with variance $\sigma_0^2$. Let $\br_0$ be the optimal \ac{mmse} estimation of $\bs_0$ from the observation $\byy_0$ with side information $\mav\brc\cdot$, i.e.,  
\begin{align}
	\br_0 = \Ex{\bs_0 \vert \byy_0 , \mav}{} = \sum_{\bs_0 \in \set{\pm1}^{K}} \bs_0 p\brc{\bs_0\vert \byy_0 , \mav}  ,
\end{align}
where the posterior distribution $p\brc{\bs_0\vert \byy_0 , \mav}$ is determined for the \textit{true} prior belief on $\bs_0$ and the \textit{true} noise variance $\sigma_0^2$. From the asymptotic results of \cite{bereyhi2022SecureLearning}, we can derive the following metrics of this setting in the asymptotic regime, i.e., $N,K\uparrow\infty$ with bounded $R = K/N$:
\begin{enumerate}
	\item The asymptotic \textit{information rate} which is defined as
		\begin{align}
		\mai\brc{\sigma_0} = \lim_{N\uparrow \infty}  \frac{1}{N}I\brc{\bs_0;\byy_0\vert \mav }.\label{eq:MI_F}
	\end{align}
\item The asymptotic joint distribution of an encoded-decoded pair, i.e., $\brc{\rms_{0k} , \rr_{0k}}$ for $k\in\dbc{K}$.
\end{enumerate}

The final expressions for these metrics are given in closed forms in terms of the so-called \textit{decoupled} setting. In the sequel, we first define the decoupled setting and then give the closed-form expression for the above metrics.

\subsection{Decoupled Setting}
Corresponding to the vector-valued setting in Section~\ref{sec:vec}, we define the decoupled setting as follows: Consider~the~uniformly distributed $s\in\set{\pm 1}$. This symbol is passed through an \ac{awgn} channel whose noise variance is controlled by the \textit{parameter} $m \in \dbc{0,1}$. Namely, the observation is given by
\begin{align}
	y\brc{m} = s + \frac{w}{\sqrt{E\brc{m}}},
\end{align}
where $w$ is zero-mean and unit-variance Gaussian noise, and
\begin{align}
E\brc{m} = \frac{\Phi' \brc{m}}{R\brc{\sigma_0^2 + \Phi\brc{1} - \Phi\brc{m} }},
\end{align}
with $\Phi\brc{\cdot}$ denoting the covariance function of $\mav\brc{\cdot}$.

In this setting, the \ac{mmse} estimator of $s$ is given by %
\begin{align}
	r\brc{m} = \tanh \brc{E\brc{m} y\brc{m} } . 
\end{align}
The input-output mutual information, i.e., 
\begin{align}
	I_{\md} \brc{m}= I \brc{s ; y\brc{m}},
\end{align}
is further given by
\begin{align}
\hspace*{-1mm}	I_{\md}\brc{m} \hspace*{-.7mm}=\hspace*{-.7mm} E\brc{m} \hspace*{-.7mm}- \hspace*{-.7mm} \Ex{\log\cosh \brc{E\brc{m} \hspace*{-.7mm}+\hspace*{-.7mm} \sqrt{E\brc{m}} w } }{w}\hspace*{-.7mm}. \hspace*{-.7mm}
\end{align}

We now define a new metric which is controlled by $m$: The \textit{energy function} is defined as
\begin{align}
\mal \brc{m} =    R I_{\md}\brc{m} + C_{\md} \brc{m} + \brc{1-m}  C'_{\md} \brc{m},
\end{align}
for the function $C_{\md}\brc{m}$ which is given by
\begin{align}
	C_{\md} \brc{m} = \mac\brc{ \frac{\Phi\brc{1} - \Phi \brc{m} }{\sigma_0^2} }.
\end{align}

\begin{proposition}
	\label{prop:1}
Let $m^\star$ be the minimizer of the energy function $\mal \brc{m}$ over $\dbc{0,1}$. Then, $\brc{\rms_{0k} , \rr_{0k}}$ for $k\in\dbc{K}$ converges~in distribution to the decoupled pair $\brc{s,r\brc{m^\star}}$, and the asymptotic information rate is given by %
	$\mai\brc{\sigma_0} = \mal \brc{m^\star}$.
\end{proposition}
\begin{proof}
	The proof is directly concluded from Corollary~1 and Proposition~2 in the extended manuscript \cite{bereyhi2022SecureLearning}.
\end{proof}

\section{All-or-Nothing Phenomenon}
\label{sec:all}
The parameter $m^\star$ in Proposition~\ref{prop:1} is often referred to as the \textit{overlap}. This appellation comes from the fact that~$m^\star$~gives the expected normalized inner product between $\bs_0$ and $\br_0$: As $m^\star$ is a minimizer of the energy function\footnote{$\mal \brc{m}$ is in fact the free energy of the corresponding spin glass; see \cite{bereyhi2022SecureLearning}.}, we have $\mal'\brc{m^\star} = 0$ as a necessary condition. This leads to
\begin{subequations}
		\begin{align}
	m^\star&=  \Ex{ \tanh \brc{ \sqrt{E\brc{m^\star}} w + E\brc{m^\star} }  }{w},\\
	& = \Ex{ r\brc{m^\star}  s }{} \stackrel{\mathrm{(a)}}{=} \Ex{\inner{\bs_0;\br_0}}{},
\end{align}
\end{subequations}
where $\mathrm{(a)}$ follows from Proposition~\ref{prop:1} considering the convergence of $\brc{\rms_{0k} , \rr_{0k}}$ to $\brc{s,r\brc{m}}$ in distribution. As the result, $m^\star$ bounds the fraction of bipolar symbols in $\bs_0$ which are correctly decoded after hard-thresholding\footnote{Which is asymptotically tight.} $\br_0$.

Proposition~\ref{prop:1} gives an interesting finding about the overlap. To illustrate this finding, we focus on the special form of~fields used in the proposed coding scheme, i.e., a Gaussian field with covariance function $\Phi\brc{u} = Pu^\lambda$. For these Gaussian fields, we use Proposition~\ref{prop:1} to investigate the behavior of $m^\star$ against $R$. This leads to the following conclusions:
\begin{itemize}
	\item For $\lambda = 1$, i.e., linear fields, the overlap starts from $m^\star=1$ at small rates, i.e., $R\downarrow 0$, and continuously decreases by growth of $R$. It however never reaches zero, i.e., $m^\star\neq 0$ for any choice of $R$.
	\item For $\lambda = 2$, i.e., purely quadratic fields, the overlap starts from $m^\star=1$ at small rates, i.e., $R\downarrow 0$, and shows a first-order phase transition at a critical rate $R^\star$ at which the overlap jumps from $m^\star = 1$ to $0\neq m^\star < 1$ discontinuously. It then shows a second-order phase transition at a threshold rate $R_{\rm Th}$ at which	$m^\star = 0$ for $R \geq R_{\rm Th}$.
	\item For $\lambda \geq 3$, the overlap shows a first-order phase transition at the critical rate $R^\star$ at which the overlap jumps from $m^\star = 1$ to $m^\star = 0$ and remains zero for $R \geq R^\star$.
\end{itemize}

The detailed proof of the above findings is out of the scope of a conference paper and can be followed in \cite[Section~IV]{bereyhi2022SecureLearning}. 

The above results lead to this conclusion: For $\lambda \geq 3$, the energy function $\mal \brc{m}$ has two local minima at $m=0$ and $m=1$ and a local maximum at some $0 < m_0 <1$. For $R < R^\star$, the global minimum is at $m=1$, leading to $m^\star=1$. This means that for $R < R^\star$, $\br_0$ perfectly recovers $\bs_0$. For $R \geq R^\star$, the overlap is zero meaning that $\br_0$ and $\bs_0$ are uncorrelated. This behavior is often called \textit{all-or-nothing} phenomenon and is observed in various inference problems, e.g., see \cite{niles2021all,barbier2020all,maillard2020phase}.

\subsection{Heuristic Derivation of the Critical Rate $R^\star$}
\label{sec:heur}
A rigorous derivation of the critical rate faces complications. However, using the all-or-noting phenomenon, the rate can be derived heuristically: Noting that the global minimum occurs at $m=0$ or $m=1$, one can compare the energy function at these two points. 

As $\Phi\brc{0} = \Phi'\brc{0}=0$, we have %
	$\mal \brc{0} =\mac\brc{ {P }/{\sigma_0^2} }$; %
however, an explicit calculation of the energy function at $m=1$ is not trivial. We hence use the following approximation: For ${x} \gg0$, we can approximately write %
	$\log \cosh\brc{x} \approx {x} - \log 2$.

Now, let us define random variable $ \tilde{w} = E\brc{1} +  \sqrt{E\brc{1}} w$. Since $\tilde{w}$ is a Gaussian random variable with mean and variance $E\brc{1}>0$, the probability of $\tilde{w}$ being close to zero or negative is negligible. We hence use the above approximation and write 
\begin{align}
	\Ex{\log \cosh\brc{\tilde{w}}}{\tilde{w}} \approx \Ex{\tilde{w}}{} - \log 2 = E\brc{1} - \log 2.
\end{align}
Hence~%
	$\mal \brc{1} \approx R\log 2$. %
This leads to the conclusion that
\begin{align}
	R^\star =   \dfrac{ 1}{\log 2} \mac\brc{ \frac{P }{\sigma_0^2} }.
\end{align}

The all-or-nothing phenomenon is further observed in terms of the information rate: At $R^\star$, the information rate changes from $\mai\brc{\sigma_0} = R\log 2$ (at $m^\star = 1$) to $\mai\brc{\sigma_0} = \mac\brc{ P /\sigma_0^2}$ (at $m^\star = 0$). In the former case, the end-to-end channel is noise-less and hence the information rate equals to the entropy rate of the message; however, by the phase transition, the end-to-end channel becomes noisy and the information rate is restricted by the channel capacity. This finding indicates that by using a \textit{strictly nonlinear} Gaussian field\footnote{More precisely, a Gaussian random field whose covariance function has a polynomial expansion with cubic and/or larger terms.} as the encoder and an \ac{mmse} decoder, the coding scheme shows a sharp phase transition at the channel capacity. This result confirms the earlier findings reported by Sourlas in \cite{sourlas1989spin} and \cite{sourlas1994spin}.

\subsection{Numerical Validations}
To validate our heuristic derivations, we numerically determine the overlap and the asymptotic information rate, given by Proposition~\ref{prop:1}, and compare them with the derivations in Section~\ref{sec:heur}. To this end, we set $\sigma_0^2 = 0.1$ and $P=1$ and plot the overlap and information rate against $R$ for various choices of $\lambda$. The results are shown in Figs.~\ref{fig:1} and \ref{fig:2}. As the figures show, for $\lambda= 3$ the first-order phase transition occurs exactly at the heuristically derived rate. By plotting the figures for larger choices of $\lambda$, one can observe that the figures are not numerically distinguishable from the one given for $\lambda=3$; see \cite{bereyhi2022SecureLearning}. We hence conjecture that for $\lambda\geq 3$ the overlap jumps at $R^\star$ from \textit{exactly} being one to \textit{exactly} being zero. A rigorous proof of this conjecture is however skipped at this point.

\section{Derivation of the Main Result}
\label{sec:Proof}
We now invoke our asymptotic derivations to sketch a proof for Result~\ref{res:1}. To this end, we consider the proposed coding scheme in Section~\ref{sec:Coding} and show that it leads to a \textit{reliable} and \textit{secure} transmission, when we set $R = K/N = C_{\rmS}/\log 2$. For brevity, we focus on scenarios with non-zero secrecy rates, i.e., $\sigma_\B^2 > \sigma_\E^2$. Noting that the proof invokes derivations based on the replica method, it is natural that it contains some heuristics. For brevity, we use the notation $C_i \hspace*{-.5mm}=\hspace*{-.5mm} \mac\brc{{P}/{\sigma_i^2}}$ for $i\in\set{\E, \B}$. 

\subsection{Proof of Reliability}
Noting that $\hat{M}$ is specified by hard-thresholding of the last $K$ symbols of $\br$, we have
	$\Pr\{M\neq \hat{M}\} \leq \Pr\set{\bts\neq \sgn{\tilde{\br}}}$.
We now consider the end-to-end channel from $\bts$ to $\sgn{\tilde{\br}}$ and model it as a \ac{bsc}~whose~flipping probability is given by
\begin{align}
	f = \frac{1}{K+\tK} \sum_{k=1}^{K+\tK} \setI \set{\tilde{\rms}_k \neq \sgn{\tilde{\rr}_k} },
\end{align}
where $\setI \set{\cdot}$ denotes the indicator function.

As we transmit \textit{uncoded} over this end-to-end channel, the transmission is considered reliable if $f=0$. As the derivation of $f$ for a particular realization of the Gaussian field is not trivial, we invoke the averaging trick and find the average of $f$ over all realizations of the Gaussian field used at the encoder. We start the derivation by noting that %
${\inner{\bts ; \sgn{\tilde{\br}}} }{} = 1- 2{f}{}$.~On the other hand, we have
\begin{align}
\inner{\bts ; \sgn{\tilde{\br}}}  =  \frac{1}{K+\tK} \sum_{k=1}^{K+\tK} \setI \set{\tilde{\rms}_k = \sgn{\tilde{\rr}_k} } - f.
\end{align}
Noting that $\abs{\tilde{\rr}_k} \leq 1$, we can write
	\begin{align}
	\inner{\bts ; \sgn{\tilde{\br}}}  
	&\geq   \frac{1}{K+\tK} \sum_{k=1}^{K+\tK} \tilde{\rms}_k {\tilde{\rr}_k} - f = 	\inner{\bts ; \tilde{\br}} - f .
\end{align}
This concludes that $f \leq 1- \inner{\bts ; \tilde{\br}}$.

From Proposition~\ref{prop:1}, we know that $\Ex{\inner{\bts ; \tilde{\br}} }{}$ in the large-system limit is determined by the overlap\footnote{This is shown at the beginning of Section~\ref{sec:all}.}. We hence conclude that in the asymptotic limit, $\Ex{f}{} \leq 1- m^\star$. As the result, $m^\star = 1$ guarantees the existence of a Gaussian field by which a reliable transmission is possible.

The all-or-nothing phenomenon of Gaussian fields with $\lambda \geq 3$ indicates that $m^\star = 1$ is achievable when %
$R+ {C_\E}/{\log 2} \leq {C_\B}/{\log 2}$.
This proves the reliability for any $R \leq C_{\rmS}/\log 2$.

\begin{figure}
	\begin{center}
		\input{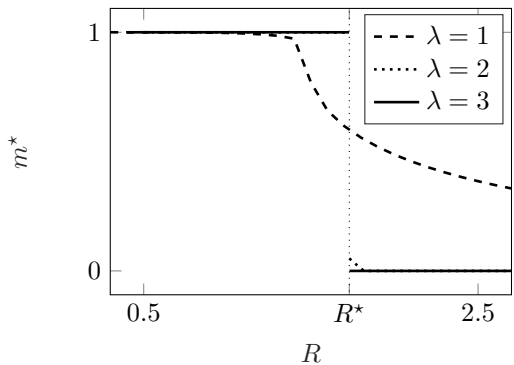}
	\end{center}
	\caption{Overlap $m^\star$ vs. $R$. The covariance function is set $\Phi\brc{u} = u^\lambda$.}
	\label{fig:1}
\end{figure}

\subsection{Proof of Security}
We start the security proof by writing\footnote{A more precise notation is to also include the random permutation in the condition, e.g., $I\brc{\bs ; \byy_\E \vert \mPi,\mav}$; however, we drop it for brevity. } 
\begin{align}
	I\brc{\bs ; \byy_\E \vert \mav}  &= I\brc{\tilde{\bs} ; \byy_\E \vert \mav}-  I\brc{\bk ; \byy_\E \vert \bs ,  \mav}. \label{eq:chain}
\end{align}
Since for $R > 0$, 
	$R+ {C_\E}/{\log 2} > {C_\E}/{\log 2}$, %
we conclude that 
\begin{align}
	\lim_{N\uparrow \infty} \frac{1}{N} I\brc{\tilde{\bs} ; \byy_\E \vert \mav} = C_\E,
\end{align}
based on the all-or-nothing phenomenon.

To calculate the second term, we first note that
\begin{align}
	I\brc{\bk ; \byy_\E \vert \bs ,  \mav} = I\brc{\bts ; \byy_\E \vert \bs ,  \mav}.
\end{align}
The right hand side is the mutual information in a genie-aided setting in which the eavesdropper recovers $\bts$ while the message $\bs$ is revealed to it: Let $\bs = \bv$ be a particular realization of the message being known to the eavesdropper. We note that after permutation there exists only one unknown bipolar symbol in each bin of $\bts$, i.e., the symbol of $\bk$ available in the bin. Thus, each bin describes a binary unknown in the inference problem. The effective transmission rate over the channel from $\bts$ to $\byy_\E$ in this case is $\tK/N = C_\E / \log 2$. We hence use Proposition~\ref{prop:1} and write\footnote{More precisely, one needs to employ an extended version of Proposition~\ref{prop:1} with a generic input distribution; however, as $\bts$ is discrete, one concludes the validity of the argument using the fact that the information rate only depends on the distribution. More discussions can be found in \cite{bereyhi2022SecureLearning}.}
\begin{align}
	\lim_{N\uparrow \infty} \frac{1}{N} I\brc{\bts ; \byy_\E \vert \bs = \bv ,  \mav} = C_\E .
\end{align}
As the right hand side does not depend on $\bv$, we have
\begin{align}
	\lim_{N\uparrow \infty} \frac{1}{N} I\brc{\bts ; \byy_\E \vert \bs  ,  \mav} = C_\E .
\end{align} 
This concludes concludes the security proof.

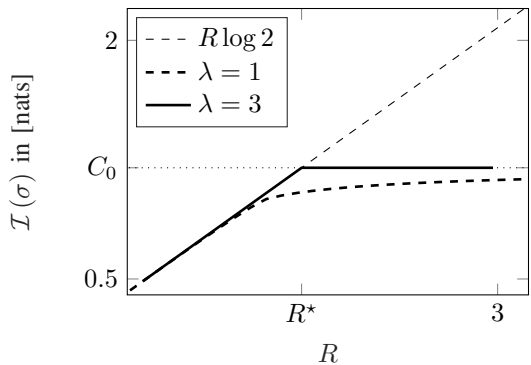
\begin{figure}
	\begin{center}
%
%
\definecolor{mycolor1}{rgb}{0.00000,0.44700,0.74100}%
\definecolor{mycolor2}{rgb}{0.85000,0.32500,0.09800}%
\definecolor{mycolor3}{rgb}{0.92900,0.69400,0.12500}%
\begin{tikzpicture}

	\begin{axis}[%
	width=2.1in,
	height=1.5in,
	at={(1.262in,0.703in)},
	scale only axis,
	xmin=.6,
	xmax=3.2,
	xlabel style={font=\color{white!15!black}},
	xlabel={$R $},
	xtick={0,1.72971580931865,3},
	xticklabels={{$0$},{$R^\star$},{$3$}},
	ymin=.4,
	ymax=2.2,
	ytick={.5,1.19894763639919,2},
	yticklabels={{$0.5$},{$C_0$},{$2$}},
	ylabel style={at={(-0.03,0.5)},font=\color{white!15!black}, },
	ylabel={$\mai\brc{\sigma}$ in [nats]},
	axis background/.style={fill=white},
	legend style={at={(0.4,.97)},legend cell align=left, align=left, draw=white!15!black}
	]
	
\addplot [color=black, dashed]
  table[row sep=crcr]{%
0.2	0.138629436111989\\
0.3	0.207944154167984\\
0.4	0.277258872223978\\
0.5	0.346573590279973\\
0.6	0.415888308335967\\
0.7	0.485203026391962\\
0.8	0.554517744447956\\
0.9	0.623832462503951\\
1	0.693147180559945\\
1.1	0.76246189861594\\
1.2	0.831776616671934\\
1.3	0.901091334727929\\
1.4	0.970406052783923\\
1.5	1.03972077083992\\
1.6	1.10903548889591\\
1.7	1.17835020695191\\
1.8	1.2476649250079\\
1.9	1.3169796430639\\
2	1.38629436111989\\
2.1	1.45560907917589\\
2.2	1.52492379723188\\
2.3	1.59423851528787\\
2.4	1.66355323334387\\
2.5	1.73286795139986\\
2.6	1.80218266945586\\
2.7	1.87149738751185\\
2.8	1.94081210556785\\
2.9	2.01012682362384\\
3	2.07944154167984\\
3.1	2.14875625973583\\
3.2	2.21807097779182\\
3.3	2.28738569584782\\
3.4	2.35670041390381\\
3.5	2.42601513195981\\
3.6	2.4953298500158\\
3.7	2.5646445680718\\
3.8	2.63395928612779\\
3.9	2.70327400418379\\
4	2.77258872223978\\
4.1	2.84190344029578\\
4.2	2.91121815835177\\
4.3	2.98053287640776\\
4.4	3.04984759446376\\
4.5	3.11916231251975\\
4.6	3.18847703057575\\
4.7	3.25779174863174\\
4.8	3.32710646668774\\
4.9	3.39642118474373\\
5	3.46573590279973\\
5.1	3.53505062085572\\
5.2	3.60436533891172\\
5.3	3.67368005696771\\
5.4	3.7429947750237\\
5.5	3.8123094930797\\
5.6	3.88162421113569\\
5.7	3.95093892919169\\
5.8	4.02025364724768\\
5.9	4.08956836530368\\
6	4.15888308335967\\
};
\addlegendentry{$R\log2$}

\addplot [color=black, dashed, line width = 1.0pt]
  table[row sep=crcr]{%
0.2	0.138629436111525\\
0.3	0.207944150627807\\
0.4	0.277258526954961\\
0.5	0.346567808970496\\
0.6	0.415848632511528\\
0.7	0.485040148269841\\
0.8	0.554033191815186\\
0.9	0.622668289330328\\
1	0.690733893987858\\
1.1	0.757954881494403\\
1.2	0.823958271235938\\
1.3	0.888182213208283\\
1.4	0.949579567363882\\
1.5	1.00197393364232\\
1.6	1.02455936361397\\
1.7	1.04076609414429\\
1.8	1.05362585501892\\
1.9	1.06425765235261\\
2	1.07327015988218\\
2.1	1.08104555028116\\
2.2	1.08784383349438\\
2.3	1.09385142925492\\
2.4	1.09920709716665\\
2.5	1.10401709613885\\
2.6	1.10836464372275\\
2.7	1.11231611643251\\
2.8	1.11592527391862\\
2.9	1.11923622806457\\
3	1.12228558449393\\
3.1	1.12510402115474\\
3.2	1.12771747384938\\
3.3	1.13014804110027\\
3.4	1.13241468467488\\
3.5	1.1345337787847\\
3.6	1.13651954551827\\
3.7	1.13838440358917\\
3.8	1.14013925023134\\
3.9	1.14179369097034\\
4	1.14335622834892\\
4.1	1.14483441803538\\
4.2	1.14623499879549\\
4.3	1.147564001359\\
4.4	1.14882684012167\\
4.5	1.15002839079484\\
4.6	1.15117305647928\\
4.7	1.15226482414824\\
4.8	1.15330731314128\\
4.9	1.15430381696885\\
5	1.15525733948963\\
5.1	1.15617062633212\\
5.2	1.1570461922802\\
5.3	1.15788634521937\\
5.4	1.15869320714102\\
5.5	1.15946873262079\\
5.6	1.16021472512102\\
5.7	1.16093285141219\\
5.8	1.16162465436353\\
5.9	1.16229156431518\\
6	1.16293490921312\\
};
\addlegendentry{$\lambda = 1$}

\addplot [color=black,dotted, forget plot]
  table[row sep=crcr]{%
0.2	1.19894763639919\\
6	1.19894763639919\\
};

\addplot [color=black, line width = 1.0pt]
table[row sep=crcr]{%
	0.7	0.485203026328893\\
	0.75	0.519860385128997\\
	0.8	0.554517743332759\\
	0.85	0.589175099807498\\
	0.9	0.623832451883395\\
	0.95	0.658489793926648\\
	1	0.69314711510415\\
	1.05	0.727804396159868\\
	1.1	0.762461605082891\\
	1.15	0.797118691627411\\
	1.2	0.831775580733311\\
	1.25	0.866432164881587\\
	1.3	0.901088296251281\\
	1.35	0.935743776897827\\
	1.371	0.95029882535362\\
	1.3735	0.952031557582671\\
	1.376	0.953764287200787\\
	1.3785	0.95549701416763\\
	1.381	0.957229738442415\\
	1.3835	0.958962459983855\\
	1.386	0.96069517875025\\
	1.3885	0.962427894699406\\
	1.391	0.964160607788703\\
	1.3935	0.965893317975009\\
	1.396	0.967626025214738\\
	1.3985	0.969358729463814\\
	1.401	0.971091430677717\\
	1.4035	0.97282412881145\\
	1.406	0.974556823819473\\
	1.4085	0.976289515655822\\
	1.411	0.978022204274026\\
	1.4135	0.979754889627124\\
	1.416	0.98148757166765\\
	1.4185	0.983220250347653\\
	1.421	0.984952925618686\\
	1.4235	0.986685597431784\\
	1.426	0.988418265737486\\
	1.4285	0.990150930485839\\
	1.431	0.991883591626362\\
	1.4335	0.993616249108062\\
	1.436	0.995348902879424\\
	1.4385	0.997081552888442\\
	1.441	0.998814199082556\\
	1.4435	1.00054684140871\\
	1.446	1.0022794798133\\
	1.4485	1.0040121142422\\
	1.451	1.00574474464077\\
	1.4535	1.00747737095378\\
	1.456	1.00920999312556\\
	1.4585	1.01094261109977\\
	1.461	1.01267522481966\\
	1.4635	1.01440783422785\\
	1.466	1.01614043926643\\
	1.4685	1.01787303987696\\
	1.471	1.01960563600042\\
	1.4735	1.02133822757724\\
	1.476	1.02307081454732\\
	1.4785	1.02480339684995\\
	1.481	1.02653597442393\\
	1.4835	1.02826854720738\\
	1.486	1.03000111513796\\
	1.4885	1.03173367815272\\
	1.491	1.0334662361881\\
	1.4935	1.03519878918003\\
	1.496	1.0369313370638\\
	1.4985	1.03866387977415\\
	1.501	1.04039641724523\\
	1.5035	1.0421289494106\\
	1.506	1.04386147620324\\
	1.5085	1.04559399755554\\
	1.511	1.04732651338558\\
	1.5135	1.0490590236496\\
	1.516	1.05079152826682\\
	1.5185	1.05252402716722\\
	1.521	1.0542565202802\\
	1.5235	1.05598900753452\\
	1.526	1.05772148885838\\
	1.5285	1.0594539641793\\
	1.531	1.06118643342424\\
	1.5335	1.06291889650412\\
	1.536	1.06465135337687\\
	1.5385	1.06638380395074\\
	1.541	1.06811624815018\\
	1.5435	1.069848685899\\
	1.546	1.07158111712043\\
	1.5485	1.07331354173701\\
	1.551	1.07504595967069\\
	1.5535	1.07677837084277\\
	1.556	1.07851077517389\\
	1.5585	1.08024317258409\\
	1.561	1.08197556299274\\
	1.5635	1.08370794631856\\
	1.566	1.08544032247962\\
	1.5685	1.08717269139336\\
	1.571	1.08890505297655\\
	1.5735	1.09063740714531\\
	1.576	1.0923697538151\\
	1.5785	1.09410209290072\\
	1.581	1.09583442431633\\
	1.5835	1.09756674797536\\
	1.586	1.09929906379064\\
	1.5885	1.1010313716743\\
	1.591	1.1027636715378\\
	1.5935	1.10449596329195\\
	1.596	1.10622824684684\\
	1.5985	1.10796052211189\\
	1.601	1.10969278899589\\
	1.6035	1.11142504740688\\
	1.606	1.11315729725225\\
	1.6085	1.11488953843868\\
	1.611	1.1166217708722\\
	1.6135	1.11835399445811\\
	1.616	1.12008620910102\\
	1.6185	1.12181841470484\\
	1.621	1.12355061117281\\
	1.6235	1.12528279840743\\
	1.626	1.12701497631052\\
	1.6285	1.12874714478318\\
	1.631	1.13047930372581\\
	1.6335	1.13221145303811\\
	1.636	1.13394359261903\\
	1.6385	1.13567572236683\\
	1.641	1.13740784217908\\
	1.6435	1.13913995195256\\
	1.646	1.14087205158342\\
	1.6485	1.14260414096698\\
	1.651	1.14433621999791\\
	1.6535	1.14606828857014\\
	1.656	1.14780034657684\\
	1.6585	1.14953239391047\\
	1.661	1.15126443046276\\
	1.6635	1.15299645612468\\
	1.666	1.15472847078647\\
	1.6685	1.15646047433765\\
	1.671	1.15819246666694\\
	1.6735	1.15992444766238\\
	1.676	1.16165641721122\\
	1.6785	1.16338837519998\\
	1.681	1.1651203215144\\
	1.6835	1.1668522560395\\
	1.686	1.1685841786595\\
	1.6885	1.1703160892579\\
	1.691	1.17204798771743\\
	1.6935	1.17377987392004\\
	1.696	1.17551174774691\\
	1.6985	1.17724360907848\\
	1.701	1.1789754577944\\
	1.7035	1.18070729377355\\
	1.706	1.18243911689403\\
	1.7085	1.18417092703319\\
	1.711	1.18590272423698\\
	1.7135	1.18763450804494\\
	1.716	1.18936627849889\\
	1.7185	1.19109803547304\\
	1.721	1.19282977884076\\
	1.7235	1.19456150847474\\
	1.726	1.19629322424674\\
	1.7285	1.19802492602783\\
	1.731	1.19894763639919\\
	1.7335	1.19894763639919\\
	1.736	1.19894763639919\\
	1.7385	1.19894763639919\\
	1.741	1.19894763639919\\
	1.7435	1.19894763639919\\
	1.746	1.19894763639919\\
	1.7485	1.19894763639919\\
	1.751	1.19894763639919\\
	1.7535	1.19894763639919\\
	1.756	1.19894763639919\\
	1.7585	1.19894763639919\\
	1.761	1.19894763639919\\
	1.7635	1.19894763639919\\
	1.766	1.19894763639919\\
	1.7685	1.19894763639919\\
	1.771	1.19894763639919\\
	1.7735	1.19894763639919\\
	1.776	1.19894763639919\\
	1.7785	1.19894763639919\\
	1.781	1.19894763639919\\
	1.7835	1.19894763639919\\
	1.786	1.19894763639919\\
	1.7885	1.19894763639919\\
	1.791	1.19894763639919\\
	1.7935	1.19894763639919\\
	1.796	1.19894763639919\\
	1.7985	1.19894763639919\\
	1.801	1.19894763639919\\
	1.8035	1.19894763639919\\
	1.806	1.19894763639919\\
	1.8085	1.19894763639919\\
	1.811	1.19894763639919\\
	1.8135	1.19894763639919\\
	1.816	1.19894763639919\\
	1.8185	1.19894763639919\\
	1.82	1.19894763639919\\
	1.87	1.19894763639919\\
	1.92	1.19894763639919\\
	1.97	1.19894763639919\\
	2.02	1.19894763639919\\
	2.07	1.19894763639919\\
	2.12	1.19894763639919\\
	2.17	1.19894763639919\\
	2.22	1.19894763639919\\
	2.27	1.19894763639919\\
	2.32	1.19894763639919\\
	2.37	1.19894763639919\\
	2.42	1.19894763639919\\
	2.47	1.19894763639919\\
	2.52	1.19894763639919\\
	2.57	1.19894763639919\\
	2.62	1.19894763639919\\
	2.67	1.19894763639919\\
	2.72	1.19894763639919\\
	2.77	1.19894763639919\\
	2.82	1.19894763639919\\
	2.87	1.19894763639919\\
	2.92	1.19894763639919\\
	2.97	1.19894763639919\\
};

\addlegendentry{$\lambda = 3$}

\end{axis}
\end{tikzpicture}%
	\end{center}
	\caption{Information rate vs. $R$ for $\Phi\brc{u} = u^\lambda$. Here, $C_0 = \mac\brc{ P /\sigma_0^2}$. }
	\label{fig:2}
\end{figure}

\section{Conclusions}
Bayesian estimation over a channel with \ac{awgn} whose generative model is described by a \textit{strictly-nonlinear} Gaussian random field shows a first-order phase transition at the capacity of the \ac{awgn} channel: By setting the rate arbitrarily close to the channel capacity, perfect recovery via \ac{mmse} estimation is asymptotically guaranteed; however, by exceeding the channel capacity, the Bayesian estimation becomes uncorrelated. This \textit{all-or-nothing} phenomenon is used to establish a secure communication in a Gaussian wiretap channel. Our investigations demonstrate that a coding scheme, whose encoder constructs the codewords by passing the message through a~Gaussian~random field, asymptotically achieves the secrecy capacity of the Gaussian wiretap channel. 

A natural direction for future work is to develop an approximate message passing algorithm for Bayesian decoding. Using such implementation, the performance of the proposed scheme can be further investigated for finite-length transmissions. The work in this direction is currently ongoing.

\bibliography{ref}

\begin{thebibliography}{10}
\providecommand{\url}[1]{#1}
\csname url@samestyle\endcsname
\providecommand{\newblock}{\relax}
\providecommand{\bibinfo}[2]{#2}
\providecommand{\BIBentrySTDinterwordspacing}{\spaceskip=0pt\relax}
\providecommand{\BIBentryALTinterwordstretchfactor}{4}
\providecommand{\BIBentryALTinterwordspacing}{\spaceskip=\fontdimen2\font plus
\BIBentryALTinterwordstretchfactor\fontdimen3\font minus
  \fontdimen4\font\relax}
\providecommand{\BIBforeignlanguage}[2]{{%
\expandafter\ifx\csname l@#1\endcsname\relax
\typeout{** WARNING: IEEEtran.bst: No hyphenation pattern has been}%
\typeout{** loaded for the language `#1'. Using the pattern for}%
\typeout{** the default language instead.}%
\else
\language=\csname l@#1\endcsname
\fi
#2}}
\providecommand{\BIBdecl}{\relax}
\BIBdecl

\bibitem{barbier2016mutual}
J.~Barbier, M.~Dia, N.~Macris, and F.~Krzakala, ``The mutual information in
  random linear estimation,'' \emph{Proc. 54th Annual Allerton Conference on
  Communication, Control, and Computing (Allerton)}, pp. 625--632, September
  2016.

\bibitem{barbier2018mutual}
J.~Barbier, N.~Macris, A.~Maillard, and F.~Krzakala, ``The mutual information
  in random linear estimation beyond iid matrices,'' \emph{Proc. IEEE
  International Symposium on Information Theory (ISIT)}, pp. 1390--1394, July
  2018.

\bibitem{bereyhi4}
A.~Bereyhi, R.~R. M{\"u}ller, and H.~Schulz-Baldes, ``Statistical mechanics of
  {MAP} estimation: General replica ansatz,'' \emph{IEEE Transations on
  Information Theory}, vol.~65, no.~12, pp. 7896--7934, August 2019.

\bibitem{bereyhi1Extension}
A.~Bereyhi and R.~R. M{\"u}ller, ``Maximum-a-posteriori signal recovery with
  prior information: Applications to compressive sensing,'' \emph{Proc. IEEE
  International Conference on Acoustics, Speech and Signal Processing
  (ICASSP)}, pp. 4494--4498, April 2018.

\bibitem{barbier2020mutual}
J.~Barbier, N.~Macris, M.~Dia, and F.~Krzakala, ``Mutual information and
  optimality of approximate message-passing in random linear estimation,''
  \emph{IEEE Transactions on Information Theory}, vol.~66, no.~7, pp.
  4270--4303, July 2020.

\bibitem{bereyhi2020thesis}
A.~Bereyhi, \emph{Statistical Mechanics of Regularized Least Squares}.\hskip
  1em plus 0.5em minus 0.4em\relax Ph.D. Dissertation, Friedrich-Alexander
  University, June 2020.

\bibitem{barbier2021performance}
J.~Barbier, W.-K. Chen, D.~Panchenko, and M.~S{\'a}enz, ``Performance of
  {Bayesian} linear regression in a model with mismatch,'' \emph{arXiv preprint
  arXiv:2107.06936}, November 2021.

\bibitem{edwards1975theory}
S.~F. Edwards and P.~W. Anderson, ``Theory of spin glasses,'' \emph{Journal of
  Physics F: Metal Physics}, vol.~5, no.~5, p. 965, May 1975.

\bibitem{mezard1987spin}
M.~M{\'e}zard, G.~Parisi, and M.~A. Virasoro, \emph{Spin glass theory and
  beyond: An Introduction to the Replica Method and Its Applications}.\hskip
  1em plus 0.5em minus 0.4em\relax World Scientific Publishing Company, 1987,
  vol.~9.

\bibitem{panchenko2013sherrington}
D.~Panchenko, \emph{The Sherrington-Kirkpatrick model}.\hskip 1em plus 0.5em
  minus 0.4em\relax Springer Science \& Business Media, 2013.

\bibitem{talagrand2006free}
M.~Talagrand, ``Free energy of the spherical mean field model,''
  \emph{Probability theory and related fields}, vol. 134, no.~3, pp. 339--382,
  May 2005.

\bibitem{sourlas1989spin}
N.~Sourlas, ``Spin-glass models as error-correcting codes,'' \emph{Nature},
  vol. 339, no. 6227, pp. 693--695, June 1989.

\bibitem{kabashima1999statistical}
Y.~Kabashima and D.~Saad, ``Statistical mechanics of error-correcting codes,''
  \emph{Europhysics Letters (EPL)}, vol.~45, no.~1, p.~97, January 1999.

\bibitem{fyodorov2019spin}
Y.~V. Fyodorov, ``A spin glass model for reconstructing nonlinearly encrypted
  signals corrupted by noise,'' \emph{Journal of Statistical Physics}, vol.
  175, no.~5, pp. 789--818, June 2019.

\bibitem{leung1978gaussian}
S.~Leung-Yan-Cheong and M.~Hellman, ``The {G}aussian wiretap channel,''
  \emph{IEEE Transactions on Information Theory}, vol.~24, no.~4, pp. 451--456,
  July 1978.

\bibitem{bereyhi2022SecureLearning}
A.~Bereyhi, B.~Loureiro, F.~Krzakala, R.~R. M\"uller, and H.~Schulz-Baldes,
  ``Bayesian inference with nonlinear generative models: Comments on secure
  learning,'' \emph{arXiv Preprint}, no. arXiv:2201.09986, Available online at
  {https://arxiv.org/abs/2201.09986}, January 2022.

\bibitem{wyner1974recent}
A.~Wyner, ``Recent results in the {Shannon} theory,'' \emph{IEEE Transactions
  on information Theory}, vol.~20, no.~1, pp. 2--10, January 1974.

\bibitem{barbier2020all}
J.~Barbier, N.~Macris, and C.~Rush, ``All-or-nothing statistical~and~com-
  putational phase transitions in sparse spiked matrix estimation,''
  \emph{arXiv Preprint}, no. arXiv:2006.07971, Available online at
  https://arxiv.org/abs /2006.07971, October 2020.

\bibitem{niles2021all}
J.~Niles-Weed and I.~Zadik, ``It was ``all'' for ``nothing'': {Sharp} phase
  transitions for noiseless discrete channels,'' \emph{arXiv Preprint}, no.
  arXiv:2102.12422, Available online at https://arxiv.org/abs/2102.12422,
  February 2021.

\bibitem{maillard2020phase}
A.~Maillard, B.~Loureiro, F.~Krzakala, and L.~Zdeborov\'{a}, ``Phase retrieval
  in high dimensions: Statistical and computational phase transitions,''
  \emph{Proc. Advances in Neural Information Processing Systems}, vol.~33, pp.
  11\,071--11\,082, December 2020.

\bibitem{liang2009information}
Y.~Liang, H.~V. Poor, S.~Shamai \emph{et~al.}, ``Information theoretic
  security,'' \emph{Foundations and Trends$^\text{\textregistered}$ in
  Communications and Information Theory}, vol.~5, no. 4--5, pp. 355--580, 2009.

\bibitem{wyner1975wire}
A.~D. Wyner, ``The wiretap channel,'' \emph{Bell system technical journal},
  vol.~54, no.~8, pp. 1355--1387, October 1975.

\bibitem{cover1975proof}
T.~Cover, ``A proof of the data compression theorem of {Slepian} and {Wolf} for
  ergodic sources (corresp.),'' \emph{IEEE Transactions on Information Theory},
  vol.~21, no.~2, pp. 226--228, March 1975.

\bibitem{yassaee2014achievability}
M.~H. Yassaee, M.~R. Aref, and A.~Gohari, ``Achievability proof via output
  statistics of random binning,'' \emph{IEEE Transactions on Information
  Theory}, vol.~60, no.~11, pp. 6760--6786, November 2014.

\bibitem{muramatsu2012construction}
J.~Muramatsu and S.~Miyake, ``Construction of codes for the wiretap channel and
  the secret key agreement from correlated source outputs based on the hash
  property,'' \emph{IEEE Transactions on Information Theory}, vol.~58, no.~2,
  pp. 671--692, February 2012.

\bibitem{sourlas1994spin}
N.~Sourlas, ``Spin glasses, error-correcting codes and finite-temperature
  decoding,'' \emph{Europhysics Letters (EPL)}, vol.~25, no.~3, p. 159, January
  1994.

\end{thebibliography}
\bibliographystyle{IEEEtran}
\end{document}